\RequirePackage{fix-cm}
\documentclass[smallextended]{svjour3}       % onecolumn (second format)
\smartqed  % flush right qed marks, e.g. at end of proof
\usepackage{graphicx}
\usepackage{amsmath}
\usepackage{amsfonts}
\usepackage{xifthen}
\usepackage{environ}
\usepackage{tabularx}
\usepackage{changepage}
\usepackage{amssymb}
\usepackage{url}

\usepackage{dashbox}
\newcommand\dboxed[1]{\dbox{\ensuremath{#1}}}

\usepackage{etoolbox}
\usepackage{caption}
\usepackage{subcaption}
\usepackage{mathtools}
\usepackage{tikz}

\usetikzlibrary{fadings}
\usetikzlibrary{shadings}
\usetikzlibrary{patterns}
\usepackage{tikz-3dplot}
\usepackage{pgfplots}
\pgfplotsset{compat=newest}
\definecolor{darkred}{RGB}{200,100,100}

\let\lat\empty
\let\set\mathbb
\def\<#1>{\langle#1\rangle}

\newcommand{\lonely}{\operatorname{lonely}}

\newtheorem{algo}{Algorithm}
\def\eatspace#1{#1}
\def\step#1#2{\par\kern1pt\dimen44=#2em\advance\dimen44 1.67em\hangindent\dimen44\hangafter=1\noindent\rlap{\small#1}\kern\dimen44\relax\eatspace}

\begin{document}

\title{Lonely Points in Simplices%\thanks{Grants or
                                                         %other notes
  % about the article that should go on the front page should be placed
  % here. General acknowledgments should be placed at the end of the article.}
}
% \subtitle{Do you have a subtitle?\\ If so, write it here}

% \titlerunning{Short form of title}        % if too long for running head

\author{Maximilian Jaroschek \and
  Manuel Kauers \and
  Laura Kov{\'a}cs
}

% \authorrunning{Short form of author list} % if too long for running head

\institute{Maximilian Jaroschek\at
  TU Wien, Institute for Logics and Computation\\
  Favoritenstr.\ 9--11,1040 Wien, Austria\\
  \email{maximilian@mjaroschek.com}\\
  Supported by the Austrian Science Fund (FWF) grant P 31427-N31.
  \and
  Laura Kov{\'a}cs\at
  TU Wien, Institute for Logics and Computation\\
  Favoritenstr.\ 9--11,1040 Wien, Austria\\
  \email{lkovacs@forsyte.at}\\
  Supported by the ERC StG SYMCAR 639270, the KAW Wallenberg Academy
  Fellowship TheProSE and the Austrian Science Fund (FWF) grant W1255-N23.
  \and
  Manuel Kauers \at
  JKU Linz, Institute for Algebra\\
  Altenberger Str.\ 69, 4040 Linz, Austria\\
  \email{manuel@kauers.de}\\
  Supported by the Austrian Science Fund (FWF) grants F5004
  and FWF P31571-N32.
}

\date{Received: date / Accepted: date}

\maketitle

\begin{abstract}
  Given a lattice $\lat L\subseteq\set Z^m$ and a subset $A\subseteq\set R^m$,
  we say that a point in $A$ is \emph{lonely} if it is not equivalent modulo $\lat L$
  to another point of~$A$.
  We are interested in identifying lonely points for specific choices of $\lat L$ when
  $A$ is a dilated standard simplex, and in conditions on $\lat L$ which ensure
  that the number of lonely points is unbounded as the simplex dilation goes to infinity.

\end{abstract}

\section{Introduction}
\label{intro}

The geometric problem considered in this article arose from an attempt to construct an
algorithm for simplifying so-called C-finite sequences. A sequence $(a_n)_{n=0}^\infty$
in the field $\set C$ of complex numbers
is called \emph{C-finite}~\cite{tetrahedron} if it satisfies a linear
recurrence with constant coefficients\footnote{W.l.o.g, we 
  consider $\set C$ instead of an algebraically closed arbitrary field of characteristic zero}, i.e., if there are constants $c_0,\dots,c_r\in\set C$,
not all zero, such that
\[
  c_0a_n + c_1a_{n+1} + \cdots + c_ra_{n+r} = 0
\]
for all $n\in\set N$. A standard example is the sequence of Fibonacci numbers (take $c_0=c_1=1$ and $c_2=-1$).
C-finite sequences and their properties  are very well understood~\cite{stanley99,everest03,tetrahedron,zeilberger13,kauers17f}.
In particular, it is known that a sequence is
C-finite if and only if it can be expressed as a linear combination of exponential terms
with polynomial coefficients, i.e., if there are polynomials $p_1,\dots,p_m\in\set C[x]$
and constants $\phi_1,\dots,\phi_m\in\set C$, such that
\[
  a_n = p_1(n)\phi_1^n + \cdots + p_m(n)\phi_m^n
\]
for all $n\in\set N$. If the $\phi_i$ are pairwise distinct and all the $p_i$ are nonzero, then
the order $r$ of the corresponding recurrence turns out to be $m+\sum_{i=1}^m\deg(p_i)$.

One of the consequences of the characterization above is that the class of C-finite sequences is
closed under addition and multiplication, i.e., when the sequences $(a_n)_{n=0}^\infty$ and $(b_n)_{n=0}^\infty$
are C-finite, then so are the sequences $(a_n+b_n)_{n=0}^\infty$ and $(a_nb_n)_{n=0}^\infty$.
In particular, when we plug a C-finite sequence into a polynomial, the result is again a C-finite
sequence. For example, since the sequence $(F_n)_{n=0}^\infty$ of Fibonacci-numbers is C-finite,
so is the sequence $(5F_n^3-7F_n^2+9F_n-4)_{n=0}^\infty$ obtained by plugging $(F_n)_{n=0}^\infty$
into the polynomial $5x^3-7x^2+9x-4\in\set C[x]$.

Given a C-finite sequence $(a_n)_{n=0}^\infty$, specified by a recurrence of order~$r$ and a set of initial
values, we want to decide whether there is a polynomial $q\in\set C[x]$ of positive degree such
that the C-finite sequence $(q(a_n))_{n=0}^\infty$ satisfies a recurrence of order less than~$r$.
This problem is of interest 
because certain number-theoretic questions about C-finite sequences
can at the moment only be answered when the recurrence order is not too large,
e.g.~see~\cite{worrell}. By using results of our paper to  pass
from $(a_n)_{n=0}^\infty$ to $(q(a_n))_{n=0}^\infty$,
we hope to extend the scope of these algorithms and advance, for
example,  their use
in applications of static analysis of computer systems, see for example~\cite{HumenbergerJK17,Worrell18}.

The construction of an algorithm for finding~$q\in\set C[x]$, such that
$(q(a_n))_{n=0}^\infty$ yields a C-finite sequence of lower order than $a$, 
has led us to the following geometric problem.
Let $\lat S\subseteq\set R^m$ be the standard simplex, i.e., the convex hull of $0$ and the unit vectors
$e_1,\dots,e_m\in\set R^m$.
Moreover, let $\lat L\subseteq\set Z^m$ be a lattice, i.e., an additive subgroup of~$\set Z^m$. Two points
$u,v\in\set R^m$ are called equivalent modulo $\lat L$ if we have $u-v\in \lat L$.
We consider the integer points in a dilation $d\lat S$ of~$\lat S$, for some $d>0$.
A point $u\in d\lat S\cap\set Z^m$ is called \emph{lonely} if there does not exist any other point
$v\in d\lat S\cap\set Z^m$ such that $u-v\in \lat L$. In this paper, we
are interested to describe properties of these lonely points. 

In Section~\ref{sec:ansatz}, we will give some more details on how the original problem about C-finite sequences
leads to the consideration of lonely points. This material is provided only as background information
and not strictly needed for the rest of the paper. In Section~\ref{sec:basics}, we summarize basic definitions
and facts about cones, simplices, and lattices, and fix the notation we use.
In Section~\ref{sec:algs} we present algorithms that for a given lattice $\lat L$ and a given $d$ determine all
the lonely points, and recognize whether the number is unbounded as $d$ goes to infinity.
Finally, in Section~\ref{sec:thms} we derive a sufficient condition on the lattice that guarantee that the number
of lonely points is unbounded.

\section{Ansatz and Exponent Lattice}\label{sec:ansatz}

Consider a C-finite sequence $(a_n)_{n=0}^\infty$ which satisfies a recurrence
of order~$r$.  We want to know whether there is a polynomial $q\in\set
C[x]\setminus\set C$ such that $(q(a_n))_{n=0}^\infty$ satisfies a recurrence of
lower order. If we have an upper bound $d$ on the degree of~$q$, then this question can
be answered as follows:
\begin{enumerate}
\item Compute $p_1,\dots,p_m\in\set C[x]$ and $\phi_1,\dots,\phi_m\in\set C$
  such that $a_n = p_1(n)\phi_1^n + \cdots + p_m(n)\phi_m^n$ for all $n\in\set N$
  (see \cite{tetrahedron} for how to do this).
\item Make an ansatz $q=q_0+q_1x+\cdots+q_dx^d$ with undetermined
  coefficients $q_0,\dots,q_d$, plug the closed form representation of step~1 into~$p$.
\item Write the resulting expression in the form
  \[
   \tikz[scale=.3,baseline=-.1cm]\draw(0,0)circle(1cm and .5cm);\psi_1^n + \cdots +
   \tikz[scale=.3,baseline=-.1cm]\draw(0,0)circle(1cm and .5cm);\psi_\ell^n
  \]
  where the $\psi_i\in\set C$ are pairwise distinct and the $\tikz[scale=.3,baseline=-.1cm]\draw(0,0)circle(1cm and .5cm);$ are polynomials
  in $n$ whose coefficients are $\set C$-linear combinations of the unknowns $q_0,\dots,q_d$.
\item For every subset $I\subseteq\{1,\dots,\ell\}$ such that $|\{1,\dots,\ell\}\setminus I|=r-1$,
  equate the coefficients with respect to $n$ in all the $\tikz[scale=.3,baseline=-.1cm]\draw(0,0)circle(1cm and .5cm);$ expressions belonging to
  some $\psi_i$ with $i\in I$ to zero and solve the resulting linear system for the
  unknowns $q_0,\dots,q_d$. If the solution space contains a vector $(q_0,\dots,q_d)$
  in which not only $q_0$ is nonzero, return the corresponding polynomial $q_0+q_1x+\cdots+q_dx^d$.
  Otherwise, try the next~$I$.
\item When no subset $I$ yields a solution, return ``there is no such $q$''.
\end{enumerate}

\begin{example}
  \begin{enumerate}
  \item The C-finite sequence $(a_n)_{n=0}^\infty$ with $a_n=1+2^n+2^{-n}$ satisfies a recurrence
    of order~3 and no lower order recurrence. With $d=2$, the algorithm sketched above finds
    the polynomial $q(x)=x^2-2x-1$. Indeed, $q(a_n)=4^n+4^{-n}$ satisfies a recurrence of order~2.
  \item The C-finite sequence $(a_n)_{n=0}^\infty$ with
    $a_n=1+3^n+3^{2 n}+2\cdot 3^{3 n}-2\cdot 3^{4 n}$
    satisfies a recurrence of order~5 and no lower order recurrence. For this input, the algorithm
    finds the polynomial $q(x)=x^2-3x+2$, and indeed,
    $q(a_n)=-3^n+7\cdot 3^{4 n}-8\cdot 3^{7 n}+4\cdot 3^{8 n}$ satisfies a recurrence of order~4.

    Similar examples can be constructed using polynomials with sparse powers. Such polynomials
    have been studied for example in~\cite{coppersmith91}.
  \item The C-finite sequence $(a_n)_{n=0}^\infty$ with $a_n=1+2^n-2^{-n}$ satisfies a recurrence
    of order~3, and with the algorithm sketched above we can show that there is no polynomial
    $q$ of degree $d\leq5$ such that $q(a_n)$ satisfies a recurrence of order~2.
  \end{enumerate}
\end{example}

When we have checked the existence of a polynomial $q$ for a specific degree $d$ and found
that no such polynomial exists, we can try again with a larger choice of~$d$. It would
be good to know when we can stop: starting from the recurrence of $(a_n)_{n=0}^\infty$, can
we determine a finite bound on the degree of the polynomials~$q$ that may lead to lower order
recurrences?

In order to see from where such a bound could emerge, restrict the search to polynomials $q$
with~$q_d=1$. Observe what happens in step~2 of
the procedure sketched above. Plugging the expression $p_1(n)\phi_1^n + \cdots + p_m(n)\phi_m^n$
into the ansatz for~$q$ produces
\begin{alignat}1 \label{eq:1}
  &q_0 + q_1 \sum_{i=1}^m p_i(n)\phi_i^n + q_2 \sum_{i,j=1}^n p_i(n)p_j(n)(\phi_i\phi_j)^n \notag\\
  &+ \cdots + \sum_{i_1,\dots,i_d=1}^n \prod_{j=1}^d p_{i_j}(n)\biggl(\prod_{j=1}^d\phi_{i_j}\biggr)^n,
\end{alignat}
so the $\psi_i$'s appearing in step~3 are precisely the products $\phi_1^{v_1}\dots\phi_m^{v_m}$
with $v_1+\cdots+v_m\leq d$. If these products 
are all distinct, then there is no way for the above
expression to vanish identically. More generally, a necessary condition for the above expression
to vanish identically for some choice of $q_0,\dots,q_{d-1}$, not all zero, is that a sufficient amount of cancellation
takes place among the various exponential sequences $((\phi_1^{v_1}\dots\phi_m^{v_m})^n)_{n=0}^\infty$.

This leads to the consideration of the so-called \emph{exponent lattice}
\[
  \lat L = \{\,(v_1,\dots,v_m)\in\set Z^m : \phi_1^{v_1}\dots\phi_m^{v_m}=1\,\}\subseteq\set Z^m,
\]
which also plays an important role for determining the algebraic relations among C-finite sequences~\cite{dependencies}.
For example, for the Fibonacci numbers, where we have $\phi_1=\frac12(1+\sqrt5)$ and $\phi_2=\frac12(1-\sqrt5)$,
the exponent lattice is generated by $(2,2)$.

A term $(\phi_1^{v_1}\dots\phi_m^{v_m})^n$ appearing in~\eqref{eq:1} cannot be canceled unless there is
some other point $(\tilde v_1,\dots,\tilde v_m)\in\set N^m$ with $\tilde v_1+\cdots+\tilde v_m\leq d$ and
$(v_1-\tilde v_1,\dots,v_m-\tilde v_m)\in\lat L$. If
$d$ is such that $r$ or more of the terms have no partner for cancellation, then it is clear that there
is no solution $q$ of degree~$d$. Moreover, if $\lat L$ is such that the number of terms without partner
tends to infinity as $d$ increases, then there is a finite bound on the degree that a solution $q$ may
have.

\section{Lattices and Cones}\label{sec:bound:lattices}\label{sec:basics}

We start by recalling  some basic concepts from discrete geometry.
Further background can be found in~\cite{beck07}, for example. 

\begin{definition}[Lattices]
  A set $\lat L\subset\set Z^m$ is called a \emph{lattice} if it contains the
  origin and for all $u,v\in\lat L$ and all $\alpha,\beta\in\set Z$ also
  $\alpha u+\beta v$ is an element of $\lat L$. For vectors
  $\ell_1,\dots,\ell_k\in\set Z^m$ we write $\< \ell_1,\dots,\ell_k>$ for the
  smallest lattice containing $\ell_1,\dots,\ell_k$, which we call generators of the
  lattice. The dimension $\dim(\lat L)$ of a lattice is defined as the dimension of
  the $\set R$-vector space it generates.
\end{definition}

We always view a lattice $\lat L\subseteq\set Z^m$ as a set of points in the
ambient space~$\set R^m$, spanned by the unit vectors $e_1,\dots,e_m$.  In
addition, it will be convenient to let $e_0$ be the zero vector.

\begin{example}
  The vectors $(3,3)$ and $(6,1)$ span a lattice in~$\set R^2$ of dimension~2.
  Some points in the lattice in the positive quadrant are depicted in
  Figure~\ref{fig:lattice:a}.
  The 2-dimensional lattice spanned by the vectors $(2,1,0)$ and $(0,2,1)$ in $\set R^3$
  is illustrated in Figure~\ref{fig:lattice:b}.
  \captionsetup[subfigure]{singlelinecheck=true}
  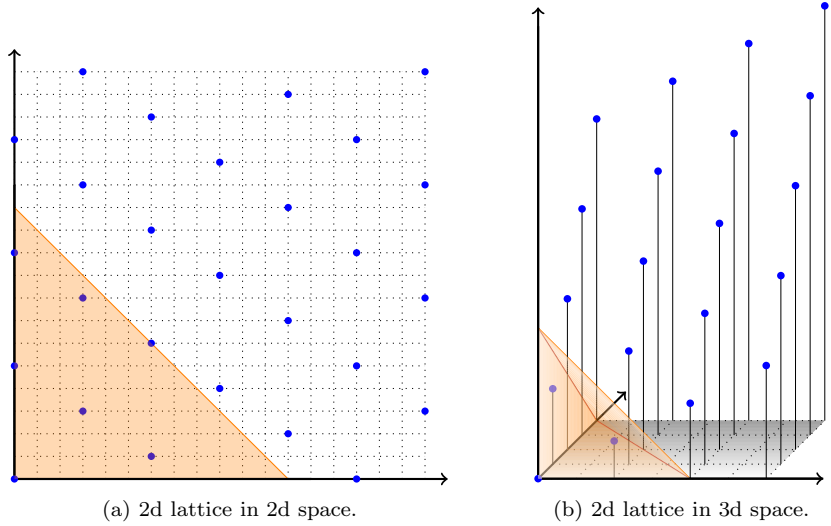
\begin{figure}[ht]
    \centering
    \begin{subfigure}[b]{.5\textwidth}
      \centering
      \vspace{1.cm}
      \newtest{\conditionb}[6]{%
        \cnttest{(#5)*(#1)+(#6)*(#3)}>{-1}
        \AND
        \cnttest{(#5)*(#2)+(#6)*(#4)}>{-1}
        \AND
        \cnttest{(#5)*(#1)+(#6)*(#3)}<{19}
        \AND
        \cnttest{(#5)*(#2)+(#6)*(#4)}<{19}
      }
      \begin{tikzpicture}[scale=0.3]
        \draw[thick,->] (0,0) -- (19,0);
        \draw[thick,->] (0,0) -- (0,19);

        \foreach \x in {0,...,18}{
          \draw[style=dotted] (\x,0) -- (\x,18);}
        \foreach \y in {0,...,18}{
          \draw[style=dotted] (0,\y) -- (18,\y);}

        \def\a{6};
        \def\b{1};
        \def\c{3};
        \def\d{3};
        \draw[color=orange] (12,0) -- (0,12);
        \foreach \i in {-10,...,10}{
          \foreach \j in {-10,...,10}{
            \ifthenelse{\conditionb{\a}{\b}{\c}{\d}{\i}{\j}}{
              \fill[color=blue] (\i*\a+\j*\c,\i*\b+\j*\d) circle (.16);}{}
          }}

        \fill[opacity=0.3, fill=orange]
        (12,0)
        -- (0,0)
        -- (0,12)
        -- cycle;
        \draw[thick,-] (0.16,0) -- (13,0);
        \draw[thick,-] (0,0.16) -- (0,4.84);
        \draw[thick,-] (0,5.16) -- (0,9.84);
        \draw[thick,-] (0,10.16) -- (0,13);

      \end{tikzpicture}
      \caption{2d lattice in 2d space.}
      \label{fig:lattice:a}
    \end{subfigure}%
    \begin{subfigure}[b]{.5\textwidth}
      \centering
      \newtest{\condition}[8]{%
        \cnttest{(#7)*(#1)+(#8)*(#4)}>{-1}
        \AND
        \cnttest{(-#7)*(#3)+(-#8)*(#6)}>{-1}
        \AND
        \cnttest{(#7)*(#2)+(#8)*(#5)}>{-1}
        \AND
        \cnttest{(#7)*(#1)+(#8)*(#4)}<{8}
        \AND
        \cnttest{(-#7)*(#3)+(-#8)*(#6)}<{5}
      }
      \begin{tikzpicture}[scale=0.5]
        \def\h{4}
        \shade[bottom color=white, top color=black!40] (0,0,0) -- (6,0,0) --
        (6,0,-4) -- (0,0,-4) -- cycle;
        \foreach \x in {0,...,6}{
          \draw[style=dotted] (\x,0,0) -- (\x,0,-4);
        }
        \foreach \z in {0,...,4}{
          \draw[style=dotted] (0,0,-\z) -- (6,0,-\z);
        }

        \draw[color=darkred!130] (0,4,0) -- (0,0,-4);
        \draw[color=darkred!130] (4,0,0) -- (0,0,-4);
        \draw[color=orange] (0,\h)--(\h,0);
        \draw[thick,->] (0,0,0) -- (7.5,0,0);
        \draw[thick,->] (0,0,0) -- (0,12.5,0);
        \draw[thick,->] (0,0,0) -- (0,0,-6);

        \def\a{2};
        \def\b{1};
        \def\c{0};
        \def\d{0};
        \def\e{2};
        \def\f{-1};

        \foreach \i in {-6,...,6}{
          \foreach \j in {-6,...,6}{
            \ifthenelse{\condition{\a}{\b}{\c}{\d}{\e}{\f}{\i}{\j}}{
              \draw[line width=0.1pt,color=black] (\i*\a+\j*\d,\i*\b+\j*\e,\i*\c+\j*\f)
              -- (\i*\a+\j*\d,0,\i*\c+\j*\f);
              \fill[color=blue] (\i*\a+\j*\d,\i*\b+\j*\e,\i*\c+\j*\f) circle (.1);}{}
          }}

        \pgfdeclarehorizontalshading{myshadingA}
        {3cm}{color(0cm)=(orange);color(0.1cm)=(orange);color(0.2cm)=(orange!70);color(0.4cm)=(orange!30);color(0.8cm)=(orange!40)}

        \shade[shading=myshadingA,shading angle=45,opacity=0.3] (\h,0) -- (0,\h) -- (1.51,1.56);
        \fill[left color=orange!20,right color=orange, opacity=0.3, fill=orange!20]
        (0,0,0)
        -- (0,0,-\h)
        -- (0,\h,0);
        \fill[bottom color=orange!20,top color=orange, opacity=0.3, fill=orange!20]
        (0,0,0)
        -- (0,0,-\h)
        -- (\h,0,0);

        \draw[thick,-] (0.1,0,0) -- (7.5,0,0);
        \draw[thick,-] (0,0.1,0) -- (0,12,0);

      \end{tikzpicture}
      \caption{2d lattice in 3d space.}
      \label{fig:lattice:b}
    \end{subfigure}
    \caption{Lattices in the positive orthant. The orange areas mark the dilated simplices
      $12\lat S$ and
      $4\lat S$ respectively.}
    \label{fig:lattice}
  \end{figure}
\end{example}

\begin{definition}[Standard Simplex]
  The \emph{standard simplex} $\lat S$ in $\set R^m$ is the convex hull of the
  points $e_0,\dots,e_m$. For $d\in\set N$, the $d$-dilation $d\lat S$ of
  $\lat S$ is the convex hull of the points $d e_0,\dots,d e_m$.
\end{definition}

We are interested in the integer points of a dilated lattice $d\lat S\subseteq\set R^m$.
Obviously, this set consists of all points $(v_1,\dots,v_m)$ in $\set Z^m$ with $v_1,\dots,v_m\geq0$
and $v_1+\cdots+v_m\leq d$. We can also describe it as an intersection of translated
cones.

\begin{definition}[Cones]
  A set $\lat C\subseteq\set Z^m$ is called a \emph{(discrete) cone} if $\lat C$
  contains the origin and we have that for all $u,v\in\lat C$ and for all
  $\alpha,\beta\in\set N$, the linear combination $\alpha u+\beta v$ is also an
  element of~$\lat C$. For vectors $c_1,\dots,c_n\in\set Z^m$ we write
  $[c_1,\dots,c_n]$ for the smallest cone containing $c_1,\dots,c_n$, which we
  call generators of the cone. For $c\in\lat C$, $[c]$ is called an \emph{edge} of
  $\lat C$ if there exists a hyperplane $H\subset \set R^m$ with
  $H\cap\lat C\subseteq [c]$. We call edges of the form $[e_i]$ or $[-e_i]$,
  $i\in\{1,\dots,m\}$ \emph{straight,} while all other edges are called
  \emph{slanted}. For $i\in\{0,\dots,m\}$, we define the $i$th \emph{corner cone}
  $\lat C_i$ of the standard simplex as
  $[e_0-e_i,e_1-e_i\dots,e_m-e_i]\subseteq\set Z^m$.
\end{definition}

Subsequently, we will only be concerned with finitely generated cones. We can
therefore assume that a cone $\lat C$ is always given as a finite set of points $c_i$,
such that for each $i$,  $[c_i]$ is an edge of $\lat C$, and for $j\neq i$ we
have $[c_i]\neq [c_j]$.

The standard simplex in $\set R^m$ has $m+1$ distinct corner cones
$\lat C_0,\dots,\lat C_m$, and the set of all integer points in $d\lat  S$,
$d\in\set N$, is equal to the intersection
$\bigcap_{i=0}^{m} (\lat C_i-de_i)$, as illustrated for
dimension~2 in Figure~\ref{fig:corner}.

\begin{figure}[th]
  \begin{subfigure}[b]{.5\textwidth}
    \centering
    \begin{tikzpicture}[scale=0.5]
      \begin{scope}
        \clip (-4.5,-4.5)
        rectangle (4.5,4.5);
        \draw[color=orange] (-5,5) -- (5,-5);
        \fill[color=orange,opacity=0.3] (0,0) -- (5,-5) -- (0,-5) -- cycle;
        \fill[color=orange,opacity=0.3] (0,0) -- (-5,5) -- (-5,0) -- cycle;
        \fill[color=orange,opacity=0.3] (0,0) -- (0,4.5) -- (4.5,4.5) -- (4.5,0) -- cycle;
      \end{scope}
      \draw[thick,->] (-5,0) -- (5,0); \draw[thick,->] (0,-5) -- (0,5);
    \end{tikzpicture}
  \end{subfigure}%
  \begin{subfigure}[b]{.5\textwidth}
    \begin{tikzpicture}[scale=0.5]
      \begin{scope}
        \clip (-4.5,-4.5) rectangle (4.5,4.5); \draw[color=orange] (0,2) -- (10,-8);
        \fill[color=orange,opacity=0.3] (0,2) -- (10,-8) -- (0,-8) -- cycle;
        \draw[color=orange] (-8,10) -- (0,2); \fill[color=orange,opacity=0.3]
        (2,0) -- (-8,10) -- (-8,0) -- cycle; \fill[color=orange,opacity=0.3]
        (0,0) -- (0,4.5) -- (4.5,4.5)-- (4.5,0) -- cycle;
      \end{scope}
      \draw[thick,->] (-5,0) -- (5,0); \draw[thick, ->] (0,-5) -- (0,5);
    \end{tikzpicture}
  \end{subfigure}
  \caption{Corner cones of the standard simplex and the intersection of the
    translated cones in $\set R^2$.}
  \label{fig:corner}
\end{figure}
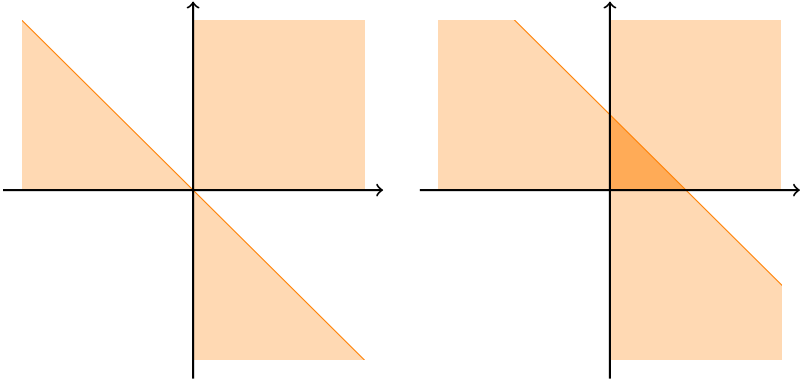

As we outlined in the earlier sections, we look for integer points in $d\lat S$
that are not connected to any other integer points in $d\lat S$ via a given
lattice $\lat L$. The next definition formalizes this idea not only for
simplices but general subsets of $\set R^m$.

\begin{definition}[Lonely Points]
  Let $\lat L\subseteq\set Z^m$ be a lattice. We define the equivalence relation
  $\sim$ on $\set Z^m$ as $u\sim v:\Leftrightarrow u-v\in\lat L$. Let $A$ be an
  arbitrary subset of $\set R^m$. A point $v\in A$ is called \emph{lonely} (with
  respect to $\lat L$), if $v\in\set Z^m$ and there is no
  $\tilde v\in (A\cap\set Z^m)\setminus\{v\}$ such that $\tilde v- v\in\lat L$.
  We write $\lonely_{\lat L}(A)$ for the set of lonely points in~$A$ and
  $\#\lonely_{\lat L}(A)\in\set N\cup\{\infty\}$ for the number of lonely points
  in~$A$.
\end{definition}

\begin{example}
  \label{ex:lonely}
  We give two examples of lattices where the number of lonely points
  in $d S$ does not grow indefinitely with~$d$.
  \begin{enumerate}
  \item\label{ex:lonely:1} For $\lat L=\<\binom{2}{-3}>\subseteq\set Z^2$ there are 9 lonely
    points in all $d\lat S$ for all $d\geq 4$ (Figure~\ref{fig:lonely},
    left).
  \item For $\lat L=\<\binom{1}{1}>\subseteq\set Z^2$ there are 4 lonely
    points in all $d\lat S$ for all $d\geq 2$.
    (Figure~\ref{fig:lonely}, right).
  \end{enumerate}
  It is easy to show that there is no lattice (other than $\{0\}$) such that the
  number of lonely points in $dS$ grows indefinitely with~$d$.
  \begin{figure}[th]
    \begin{subfigure}[b]{.5\textwidth}
      \centering
      \begin{tikzpicture}[scale=.67]
        \draw[thick,->] (0,0)--(6,0);
        \draw[thick,->] (0,0)--(0,6);
        \foreach \x in {0,...,5}
        \foreach \y in {0,...,5}
        \fill[color=blue] (\x,\y) circle (2.23pt);
        \fill[opacity=0.3,fill=orange] (0,0)--(5,0)--(0,5)--cycle;
        \foreach \x/\y in {0/0, 0/1, 0/2, 1/0, 1/1, 1/2,3/2, 4/1, 5/0} \draw[style=thick] (\x,\y) circle (3.7pt);
        \clip (0,0) rectangle (5.5,5.5);
        \foreach \y in {0,...,14}
        \draw[opacity=0.3] (0,\y) -- (\y*2,\y-\y*3)  (0,\y+0.5) -- (\y*2,\y+0.5-\y*3);
      \end{tikzpicture}
    \end{subfigure}%
    \begin{subfigure}[b]{.5\textwidth}
      \centering
      \begin{tikzpicture}[scale=.67]
        \draw[thick,->] (0,0)--(6,0);
        \draw[thick,->] (0,0)--(0,6);
        \foreach \x in {0,...,5}
        \foreach \y in {0,...,5}
        \fill[color=blue] (\x,\y) circle (2.23pt);
        \fill[opacity=0.3,fill=orange] (0,0)--(5,0)--(0,5)--cycle;
        \foreach \x/\y in {0/5, 0/4, 4/0, 5/0} \draw[style=thick] (\x,\y) circle (3.7pt);
        \clip (0,0) rectangle (5.5,5.5);
        \foreach \x in {-8,...,7}
        \draw[opacity=0.3] (0,\x)--(8,\x+8);
      \end{tikzpicture}
    \end{subfigure}
    \caption{Illustration of Example~\ref{ex:lonely}. Lonely points are encircled.}
    \label{fig:lonely}
  \end{figure}
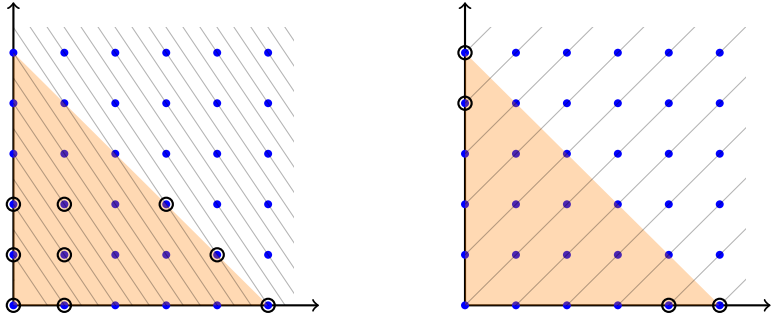
\end{example}

\begin{example}
  Let $\lat L\subseteq\set Z^4$ be the lattice generated by the
  vectors $(2,0,-1,0)$ and $(1,1,0,-1)$.
  Then there are infinitely many lonely points in any corner cone.
  For example, for each $i=0,\dots,4$, all vectors of the form $(0,n,0,0)-de_i$ with
  $d\geq n\geq 0$ are
  lonely in $\lat C_i$.
\end{example}

\makeatletter
\pgfdeclarepatternformonly[\LineSpace]{mygrid}{\pgfqpoint{-1pt}{-1pt}}{\pgfqpoint{\LineSpace}{\LineSpace}}{\pgfqpoint{\LineSpace}{\LineSpace}}%
{
  \pgfsetcolor{\tikz@pattern@color}
  \pgfsetlinewidth{0.3pt}
  \pgfpathmoveto{\pgfqpoint{0pt}{0pt}}
  \pgfpathlineto{\pgfqpoint{\LineSpace + 0.1pt}{\LineSpace + 0.1pt}}
  \pgfusepath{stroke}
}
\makeatother

\newdimen\LineSpace
\tikzset{
  LineSpace/.code={\LineSpace=#1},
  LineSpace=3pt
}

Our goal is to count the lonely points in a dilated simplex. As we will use the
translated corner cones to characterize the points inside of a dilated simplex,
we want to make sure that lonely points stay lonely after any translation.

\begin{lemma}
  \label{lem:translatelonely}
  Let $\lat L\subset\set Z^m$ be a lattice and let $v\in A\subseteq\set
  R^m$. If $v\in\lonely_{\lat L}(A)$, then $v+t\in\lonely_{\lat L}(A+t)$ for any $t\in\set Z^m$.
\end{lemma}

\begin{proof}
  Suppose $v+t\notin \lonely_{\lat L}(A+t)$. Then there exists a $\tilde v\in A$
  such that $(v+t)\sim (\tilde v+t)$. It follows that $v-\tilde v = (v
  +t) - (\tilde v +t)\in\lat L$, so $v\sim \tilde v$. \qed
\end{proof}

\section{Counting and Identifying Lonely Points}\label{sec:general:lonely}\label{sec:algs}

In this section we develop algorithms for deciding whether in a given
setting the number of lonely points is finite or infinite, as well as an algorithm
which in the finite case determines how many lonely points there are. First we characterize
loneliness of points in cones, and then we relate the loneliness of points in a
dilated cone $d\lat S$ to the loneliness of points in its corner cones.

\begin{lemma}
  \label{lem:cones}
  Let $\lat L\subseteq\set Z^m$ be a lattice and
  $\lat C\subseteq\set Z^m$ be a cone.
  \begin{enumerate}
  \item\label{lem:cones:1} If $\lat C$ has any lonely points, then $0$ is one of
    them.
  \item\label{lem:cones:2} $\lat C$ has lonely points if and only if $\lat
    L\cap\lat C=\{0\}$.
  \item\label{lem:cones:3} If $u\in\lat C$ is not lonely,
    then also $u+v$ is not lonely for any $v\in\lat C$.
  \end{enumerate}
\end{lemma}

\begin{proof}
  \begin{enumerate}
  \item If $0$ is not lonely, it is equivalent to some other point of $\lat C$,
    say to $u\neq0$.  Then $u=u-0\in\lat L$. Let $v$ be an arbitrary element
    of~$\lat C$. Since $u\in\lat C$, we have $v+u\in \lat C$, and since $v$ and
    $v+u$ are equivalent, $v$ is not lonely.
  \item If $\lat C$ has lonely points, then, by the previous item, $0$ is one of
    them, hence $\lat L\cap \lat C=\{0\}$. For the other direction, if
    $\lat L\cap \lat C=\{0\}$, then $0$ is lonely.
  \item If $u$ is not lonely, then there exists
    $\tilde u\in \lat C\setminus\{u\}$ with $u\sim \tilde u$.  Then also
    $u+v\sim \tilde u+v$, and since $\tilde u+v$ is in $\lat C$ and different
    from $u+v$, the claim follows.\qed
  \end{enumerate}
\end{proof}

\begin{proposition}
  \label{prop:infinite}
  Let $\lat L\subseteq\set Z^m$ be a lattice and
  $\lat C=[c_1,\dots,c_n]\subseteq\set Z^m$ be a cone.
  \begin{enumerate}
  \item\label{prop:infinite:1} If $\lat C$ has infinitely many lonely points,
    then there is an $i\in\{1,\dots,n\}$ such that all points in $[c_i]$ are
    lonely in $\lat C$.
  \item\label{prop:infinite:2} Let $i\in\{1,\dots,n\}$.  Then all points in
    $[c_i]$ are lonely in $\lat C$ if and only if $\lat L\cap \lat C=\{0\}$ and
    $(\lat L+\<c_i>)\cap \lat C=[c_i]$.
  \end{enumerate}
\end{proposition}

\begin{proof}
  \begin{enumerate}
  \item Suppose to the contrary all edges $[c_i]$ contain a nonlonely point,
    say $\alpha_1c_1,\dots,\alpha_nc_n$ are not lonely for certain positive
    integers $\alpha_1,\dots,\alpha_n$.  By part~\ref{lem:cones:3} of
    Lemma~\ref{lem:cones}, all points $\beta_1c_1+\cdots+\beta_nc_n$ with
    $\beta_1\geq\alpha_1,\dots,\beta_n\geq\alpha_n$ are not lonely.  Thus there
    remain only finitely many candidates for lonely points.
  \item ``$\Rightarrow$''\quad If all points in $[c_i]$ are lonely, then
    $\lat C$ has lonely points, so $\lat L\cap \lat C=\{0\}$ by
    part~\ref{lem:cones:2} of Lemma~\ref{lem:cones}.  It remains to show that
    $(\lat L+\<c_i>)\cap \lat C=[c_i]$.  The direction ``$\supseteq$'' is clear.
    To show ``$\subseteq$'', let $v\in(\lat L+\<c_i>)\cap \lat C$, say
    $v=\ell+\alpha c_i\in \lat C$ for some nonzero $\ell\in\lat L$ and
    $\alpha\in\set Z$. If $\alpha> 0$, then $v\sim \alpha c_i$, in contradiction
    to the loneliness of $\alpha c_i$. Otherwise, for $\alpha \leq 0$, we have
    $\ell=v+(-\alpha) c_i\in \lat C$, a contradiction to
    $\lat L\cap\lat C=\{0\}$.

    ``$\Leftarrow$''\quad Assume $u=\alpha_ic_i$ is not lonely, say $u\sim v$
    for some $v\in \lat C\setminus\{u\}$.  Then $u-v\in\lat L$ implies
    $v\in\lat L+\<c_i>$, so $v\in [c_i]$, say $v=\beta_ic_i$ for some
    $\beta\in\set N\setminus\{\alpha_i\}$.  But then
    $0\neq u-v\in\lat L\cap \lat C=\{0\}$, a contradiction.\qed
  \end{enumerate}
\end{proof}

The conditions of Proposition~\ref{prop:infinite} give rise to the following algorithm
for deciding whether a cone contains infinitely many lonely points.

\begin{algo}\label{alg:infiniteQ} (hasInfinitelyManyLonelyPoints)
  \par\noindent\textbf{Input:} a lattice $\lat L\subseteq\set Z^m$, a cone $\lat C=[c_1,\dots,c_n]\subseteq\set Z^m$
  \par\noindent\textbf{Output:} true or false, depending on whether $\lat C$ contains infinitely many lonely points

  \par\medskip
  \step 10 if $\lat L\cap\lat C\neq\{0\}$ then return false
  \step 20 for $i=1,\dots,n$, do:
  \step 31 if $(\lat L+\<c_i>)\cap\lat C=[c_i]$ then return true
  \step 40 return false
\end{algo}

The tests in lines 1 and~3 can be performed using integer linear programming~\cite{schrijver98}.
If $\lat L=\<\ell_1,\dots,\ell_k>=[\ell_1,\dots,\ell_k,-\ell_1,\dots,-\ell_k]$, we can find
nonnegative integers $\alpha_1,\dots,\alpha_k,\alpha_{-1},\dots,\alpha_{-k},\beta_1,\dots,\beta_n$
such that
\[
 (\alpha_1-\alpha_{-1})\ell_1+\dots+(\alpha_k-\alpha_{-k})\ell_k=\beta_1c_1+\dots+\beta_n c_n
\]
and such that $\beta_1+\cdots+\beta_n$ is maximized. We have $\lat L\cap\lat C=\{0\}$
if and only if the optimal solution is $\beta_1=\cdots=\beta_n=0$.

Similarly, in order to check whether $(\lat L+\<c_i>)\cap\lat C=[c_i]$, we can find
nonnegative integers $\alpha_1,\dots,\alpha_k,\alpha_{-1},\dots,\alpha_{-k},\gamma_1,\gamma_{-1},\beta_1,\dots,\beta_n$
such that
\[
 (\alpha_1-\alpha_{-1})\ell_1+\dots+(\alpha_k-\alpha_{-k})\ell_k+(\gamma_1-\gamma_{-1})c_i=\beta_1c_1+\dots+\beta_n c_n
\]
and $\beta_1+\cdots+\beta_{i-1}+\beta_{i+1}+\cdots+\beta_n$ is maximized.
If the intersection $[c_i]\cap[c_1,\dots,c_{i-1},c_{i+1},\dots,c_n]$ only
contains $0$, then
$(\lat L+\<c_i>)\cap\lat C$ is contained in $[c_i]$ if and only if the optimal solution is
$\beta_1=\cdots=\beta_n=0$.
In our setting, we can always assume that $c_1,\dots,c_n$ are linearly independent over~$\set Q$,
and in this case, the condition $[c_i]\cap[c_1,\dots,c_{i-1},c_{i+1},\dots,c_n]=\{0\}$
is always satisfied.

When there are only finitely many lonely points, we can next determine how many there are.
Part~\ref{lem:cones:3} of Lemma~\ref{lem:cones} says that when some $v\in\lat C$ is not lonely, then no
point in the translated cone $v+\lat C$ is lonely either. It follows from Dickson's lemma
(\cite{becker93}, see also Lemma~4 of~\cite{aparicio12}) that the set of nonlonely points in $\lat C$
is in fact a finite union of such translated cones~$v+\lat C$, quite similar to the
leading-term ideals in Gr\"obner basis theory~\cite{becker93,cls,buchberger10}. Inspired by the
FGLM-algorithm from that theory~\cite{faugere93,cls}, we arrive at the following algorithm for counting the number of
lonely points in a cone.

\begin{algo}\label{alg:fglm} (numberOfLonelyPoints)
  \par\noindent\textbf{Input:} a lattice $\lat L\subseteq\set Z^m$ and a cone $\lat C=[c_1,\dots,c_n]\subseteq\set Z^m$
  \par\noindent\textbf{Output:} $\#\lonely_{\lat L}(\lat C)$

  \par\medskip
  \step 10 if $\#\lonely_{\lat L}(\lat C)=\infty$, return $\infty$\quad(( using Algorithm~\ref{alg:infiniteQ} ))
  \step 20 if $0$ is not lonely, return $0$
  \step 30 $\mathrm{todo} = \{e_1,\dots,e_n\}\subseteq\set R^n$\quad(( list of unit vectors of length~$n$ ))
  \step 40 $B = \emptyset$\quad(( collected nonlonely points ))
  \step 50 $npoints=1$\quad(( number of lonely points seen so far ))
  \step 60 while $|\mathrm{todo}|>0$, do:
  \step 71 select an element $v=(v_1,\dots,v_n)$ with $||v||_1$ minimal from $\mathrm{todo}$
  \step 81 $\mathrm{todo} = \mathrm{todo}\setminus\{v\}$
  \step 91 if $v_1c_1+\dots+v_nc_n$ is a lonely point, then:
  \step {10}2 $npoints=npoints+1$
  \step {11}2 for $i=1,\dots,n$, do:
  \step {12}3 if $\forall b\in B:v_1c_1+\cdots+v_nc_n+c_i\not\in b+\mathcal{C}$, then
  \step {13}4 $\mathrm{todo}=\mathrm{todo}\cup\{v+e_i\}$
  \step {14}1 else\quad(( $v$ is not lonely ))
  \step {15}2 $B=B\cup\{v_1c_1+\cdots+v_nc_n\}$
  \step {16}0 return $npoints$
\end{algo}

Three aspects need to be discussed in order to justify this algorithm: (1) that all indicated operations can
be performed algorithmically, (2) that it returns the correct output, and (3) that it terminates
for every input. Concerning the first point, the only questionable steps are the checks in steps~2
and~9 whether a given point is lonely. In order for $v$ to be not lonely, there
must be integers $\alpha_1,\dots,\alpha_k$, not all zero, such that
$v+\alpha_1\ell_1+\cdots+\alpha_k \ell_k$
also belongs to~$\lat C$, where $\ell_1,\dots,\ell_k$ are generators of $\lat L$.
Whether such integers exist can be determined with linear programming.

For the correctness,
observe first that the output $npoints$ is a lower bound on the number of lonely points, because the counter
is only incremented when we have found a new lonely point~$v$. Since we always consider the candidate
of least norm and in line~13 always add elements of larger norm to the todo-list, it is excluded that
we count the same point more than once. In order to see that the output is also an upper bound, observe
that part~\ref{lem:cones:3} of Lemma~\ref{lem:cones} implies that when $b$ is not lonely, then all the points in
$b+\lat C$ are not lonely either, so it is fair to exclude them from consideration in step~12. Since
all other points will be considered, there is no danger of undercounting. This establishes the correctness.

Finally, for justifying the termination, observe that the number of iterations of the main loop is
bounded by the number of lonely points plus the number of points that are not lonely but also not
contained in a translated cone $b+\lat C$ where $b$ is a nonlonely point discovered earlier.
By line~1, the number of lonely points is finite when the algorithm reaches the main loop, and we have
already argued above that the number of nonlonely points not contained in a translated
cone rooted at an earlier discovered nonlonely point is finite as well.

\begin{example}
  Consider the lattice $\lat L=\<\binom{2}{-3}>\subseteq\set Z^2$ and the cone
  $\lat C=[e_1,e_2]\subseteq\set Z^2$.  This cone is the corner cone $C_0$ in
  the situation considered in Example~\ref{ex:lonely}, part~\ref{ex:lonely:1}
  and depicted in Figure~\ref{fig:lonely}.  Algorithm~\ref{alg:fglm} identifies the
  lonely points of $\lat C$ as follows.
  \begin{center}
    \begin{tabular}{c|cccc|r}
      iteration & $v$ & todo & $B$ & $npoints$ & comment \\\hline
      0 & & $\{\binom10,\binom01\}$ & $\emptyset$ & $1$ & initialization \\[2pt]
      1 & $\binom10$ & $\{\binom20,\binom11,\binom01\}$ & $\emptyset$ & $2$ & $v$ is lonely \\[2pt]
      2 & $\binom01$ & $\{\binom20,\binom11,\binom02\}$ & $\emptyset$ & $3$ & $v$ is lonely \\[2pt]
      3 & $\binom20$ & $\{\binom11,\binom02\}$ & $\{\binom20\}$ & $3$ & $v$ is not lonely \\[2pt]
      4 & $\binom11$ & $\{\binom12,\binom02\}$ & $\{\binom20\}$ & $4$ & $v$ is lonely \\[2pt]
      5 & $\binom02$ & $\{\binom12,\binom03\}$ & $\{\binom20\}$ & $5$ & $v$ is lonely \\[2pt]
      6 & $\binom12$ & $\{\binom13,\binom03\}$ & $\{\binom20\}$ & $6$ & $v$ is lonely \\[2pt]
      7 & $\binom03$ & $\{\binom13\}$ & $\{\binom20,\binom03\}$ & $6$ & $v$ is not lonely \\[2pt]
      8 & $\binom13$ & $\emptyset$ & $\{\binom20,\binom03,\binom13\}$ & $6$ & $v$ is not lonely
    \end{tabular}
  \end{center}
\end{example}

The next proposition connects the lonely points in a simplex to the lonely points in its corner cones.

\begin{proposition}\label{prop:cornerlonely}
  Let $\lat L\subseteq\set Z^m$ be a lattice and let $\lat S\subseteq\set R^m$ be the standard simplex.
  \begin{enumerate}
  \item A corner $de_i$ of $d\lat S$ is lonely for all sufficiently
    large $d\in\set N$ iff $0$ is a lonely point of the corresponding
    corner cone $\lat C_i$.
  \item\label{prop:cornerlonely:2} $\forall d\in\set N: \lonely_{\lat L}(d\lat S) \supseteq \bigcup \{v-de_i\mid\exists i>0:
    v\in\lonely_{\lat L}(\lat C_i)\}\cap \set N^m$.
  \item\label{prop:cornerlonely:3} The following are equivalent:
    \begin{enumerate}
    \item $\forall i:\#\lonely_{\lat L}(\lat C_i)=\infty.$
    \item $\exists i: \#\lonely_{\lat L}(\lat C_i)=\infty,$
    \item $\forall r\in\set N\;\exists d\in\set N: \#\lonely_{\lat L}(d\lat S)>r,$
    \end{enumerate}
  \end{enumerate}
\end{proposition}

\begin{proof}
  \begin{enumerate}
  \item
    Let $d e_i$ be a corner of $d\lat S$, and suppose $d$ is large.

    ``$\Rightarrow$''\quad We show: if $0$ is not a lonely point of the corner
    cone $\lat C_i=[e_0-e_i,\dots,e_m-e_i]$, then $d e_i$ is not a lonely point of
    $d S$.  If $0$ is not a lonely point of the corner cone, the corner
    cone contains some nonzero element of~$\lat L$, say
    $\ell=\alpha_0(e_0-e_i)+\cdots+\alpha_m(e_m-e_i)\in \lat L$ for certain
    $\alpha_0,\dots,\alpha_m\in\set N$.  Assuming, as we may, that
    $d>\alpha_0+\cdots+\alpha_m$, we have that $d e_i+\ell$ is an
    interior point of $d\lat S$ which is equivalent to $d e_i$, proving
    that $d e_i$ is not lonely.

    ``$\Leftarrow$''\quad We show: if
    $d e_i$ is not a lonely point of $d S$, then $0$ is not a lonely
    point of the corner cone. Indeed, suppose that $d e_i$ is equivalent to
    another point $v$ of $d\lat S$, say to $v=\beta_1e_1+\cdots+\beta_me_m$ for
    some $\beta_1,\dots,\beta_m\geq0$ whose sum is at most~$d$.  Then
    $v-de_i=\beta_1(e_1-e_i)+\cdots+\beta_m(e_m-e_i)+(d-\sum_j\beta_j)(e_0-e_i)$
    belongs to the $i$th corner cone, so $0$ is not a lonely point of that cone.

  \item Denote the set on the right hand side by $A_d$. Then
    $A_d\subset \lonely_{\lat L}(d\lat S)$ holds for any $d$: If
    $v-de_i\in\set N^m$ is such that $v$ is lonely in $\lat C_i$, then by Lemma~\ref{lem:translatelonely}
    $v-de_i$ is lonely in $\lat C_i-d e_i$, which contains $d\lat S$.
  \item
    ``(a) $\Rightarrow$ (b)''\quad is trivial.

    ``(b) $\Rightarrow$ (c)''\quad is an immediate consequence of part~\ref{prop:cornerlonely:2}.

    ``(c) $\Rightarrow$ (a)''\quad Suppose that $\lonely_{\lat L}(\lat C_i)$
    only contains finitely many elements for some corner cone $\lat C_i=[c_1,\dots,c_n]$. Then,
    by part~\ref{prop:infinite:1} of Proposition~\ref{prop:infinite}
    there exists a $d'$ such that for every edge $[c_j]$ in $\lat C_i$ the point
    $d' c_j$ is not lonely. For each such edge we let $d_{j}$ be the minimal
    euclidean distance of $d'c_j$ to some other element in $\lat C_i$ equivalent
    to $d'c_j$. Then any point $v=\sum \alpha_jc_j$ in $\lat C_i$ is equivalent
    to some point in distance $d_{j}$ if $\alpha_j\geq d'$ for some $j$. Setting
    $d$ to be the maximum of the $d_{j}$ this means that every such $v$ is
    equivalent to some point in distance $\leq d$. Then a point $v$ in
    $\tilde d\lat S$ for $\tilde d\geq d$ is lonely only if the coordinates of
    $v-\tilde de_i$ with respect to the generators $c_j$ of the $i$th corner cone are
    bounded by $d$, leaving only finitely many possible values for $v-\tilde de_i$.\qed
  \end{enumerate}
\end{proof}

For a specific $d\in\set N$, there are only finitely many points in $d\lat S$, and
for each of them, we can decide whether it is lonely in a similar way as described
above for a given point in a cone. The issue reduces to a linear programming question.
What we are interested in is how far the number of lonely points can grow as $d$
increases. Proposition~\ref{prop:cornerlonely} says that the lonely points in $d\lat S$
for sufficiently large $d$ are essentially the lonely points of the corner cones.
When a cone has only finitely many lonely points, they are all clustered near the apex,
so as soon as $d$
is sufficiently large, the number of lonely points in the dilated simplex $d\lat S$
is exactly the sum of the number of lonely points in its corner cones. When at least
one corner cone has infinitely many lonely points, then the number of lonely points
in $d\lat S$ is unbounded as $d$ goes to infinity. In summary, we obtain the following
algorithm.

\begin{algo} (ultimateNumberOfLonelyPoints)
  \par\noindent\textbf{Input:} a lattice $\lat L\subseteq\set Z^m$
  \par\noindent\textbf{Output:} $\lim_{d\to\infty}\#\lonely_{\lat L}(d\lat S)$

  \par\medskip
  \step 10 $s=0$
  \step 20 for all $i\in\{0,\dots,m\}$, do:
  \step 31 $\lat C=[e_0-e_i,\dots,e_m-e_i]$\quad(( consider the $i$th corner cone ))
  \step 41 $s=s+\#\lonely_{\lat L}(\lat C)$\quad(( use Algorithm~\ref{alg:fglm} ))
  \step 50 return $s$
\end{algo}

We have implemented the algorithms described in this section in Mathematica. The code is available on
the personal website of the second author.

\section{Lonely Points for Small Lattices}\label{sec:general:smalllat}\label{sec:thms}

It is clear that all integer points in $dS$ are lonely when $\lat L=\{0\}$ and
that there are no lonely points when $\lat L=\set Z^m$. More generally,
geometric intuition suggests that there should be more lonely points when
$\lat L$ is ``small''. The main result of the present section makes this
intuition quantitative. We show that whenever the dimension of $\lat L$ is less
than a certain constant multiple of the ambient dimension~$m$, then there is a
corner cone which satisfies the conditions of part~\ref{prop:infinite:2} of
Proposition~\ref{prop:infinite} and thus has infinitely many lonely points.

In the subsequent proofs we make use of sign vectors and equations. The possible
components of a sign vector are $+$, $-$, $\oplus$, $\ominus$ or $0$. We can
assign a sign vector $s$ to a given $v\in\set R^m$ in the following way. If the
$i$th component of $v$ is nonnegative, then the $i$th component of $s$ is $+$ or
$\oplus$. If the $i$th component of $v$ is nonpositive, then the $i$th component
of $s$ is $-$ or $\ominus$. If a component of $v$ is zero, then the
corresponding component of $s$ can be $0,+,-,\oplus$ or $\ominus$. A component
of $s$ is $\oplus$ or $\ominus$ only if the absolute value of the corresponding
component of $v$ is greater than or equal to the sum of the absolute values of
all other components.  With these rules, any equation $s_1+\dots+s_k=s$ of sign
vectors $s_1,\dots,s_{k},s$ is a valid equation if there are vectors
$v_1,\dots,v_{k},v\in\set R^m$, such that $v_1+\dots+v_k=v$ and for each
$i=1,\dots,k$, $s_i$ is a valid sign vector for $v_i$ and $s$ is a valid sign
vector for $v$.

\begin{example}
  \label{ex:signs}
  For the equation
  \[
    \begin{pmatrix}
      -2 \\ 1 \\ 0
    \end{pmatrix}
    +
    \begin{pmatrix}
      0 \\ -1 \\ 0
    \end{pmatrix}
    +
    \begin{pmatrix}
      1 \\ 1 \\ -1
    \end{pmatrix}
    =
    \begin{pmatrix}
      -1 \\ 1 \\-1
    \end{pmatrix},
  \]
  two valid sign equations are
  \[
    \begin{pmatrix}
      \ominus \\ + \\ -
    \end{pmatrix}
    +
    \begin{pmatrix}
      + \\ \ominus \\ -
    \end{pmatrix}
    +
    \begin{pmatrix}
      + \\ + \\ -
    \end{pmatrix}
    =
    \begin{pmatrix}
      - \\ + \\-
    \end{pmatrix},
    \text{ and }
        \begin{pmatrix}
      - \\ + \\ 0
    \end{pmatrix}
    +
    \begin{pmatrix}
      0 \\ - \\ 0
    \end{pmatrix}
    +
    \begin{pmatrix}
      + \\ + \\ -
    \end{pmatrix}
    =
    \begin{pmatrix}
      - \\ + \\-
    \end{pmatrix}.
  \]
\end{example}

We use a shorthand matrix notation
\[\begin{bmatrix}
    s_1 & s_2 & \cdots & s_k
  \end{bmatrix}=s,\] for the sign equation $s_1+\dots+s_k=s$, with
the square brackets indicating that the columns of the matrix are summed up to
obtain the right hand side. To further shorten notation, we use $\boxed{\oplus}$ and
$\boxed{\ominus}$ for nonempty square blocks of the form
\[
  \begin{matrix}
    \oplus & -      & \cdots & \cdots & \cdots       & -      & -      \\
    -      & \oplus &  &&       & -      & -      \\
    -      & -      & \ddots &&       & -      & -      \\
    \vdots & \vdots &   &\ddots&     & \vdots & \vdots \\
    -      & -      &    && \ddots   & -      & -      \\
    -      & -      &    &&     & \oplus & -      \\
    -      & -      & \cdots & \cdots & \cdots       & -      & \oplus
  \end{matrix}\qquad \text{ and }\qquad
  \begin{matrix}
    \ominus & +      & \cdots & \cdots & \cdots       & +      & +      \\
    +      & \ominus &  &&      & +      & +      \\
    +      & +      &   \ddots&&     & +     & +      \\
    \vdots & \vdots &   &\ddots&     & \vdots & \vdots \\
    +      & +      &   &&\ddots     & +      & +      \\
    +      & +      &   &&     & \ominus & +      \\
    +      & +      & \cdots & \cdots & \cdots       & +      & \ominus

  \end{matrix}
\]
respectively, where the number of rows/columns is either clear from the context
or irrelevant. Similarly we use $\raisebox{0pt}[10pt][6pt]{$\dboxed{+}$}$,
$\dboxed{-}$, and
$\dboxed{\hspace{1.4pt}\raisebox{0pt}[6pt][1pt]{$0$}\hspace{1.4pt}}$ for blocks
that only contain $+$, $-$, or $0$ respectively, with the difference that these
blocks do neither have to be square blocks nor nonempty.

\begin{example}
  The first sign equation in Example~\ref{ex:signs} can be written as
  \[\begin{bmatrix}
      \raisebox{0pt}[11pt][0pt]{$\boxed{\ominus}$} \\[4pt]
      \raisebox{0pt}[0pt][7pt]{$\dboxed{-}$}
    \end{bmatrix}
    +
    \begin{pmatrix}
      + \\ + \\ -
    \end{pmatrix}
    =
    \begin{pmatrix}
      - \\ + \\-
    \end{pmatrix}.   
  \]
\end{example}
For any vector $v$ in $\set R^m$ we define the \emph{balance} $\tau(v)$ of $v$
to be the sum of the components of $v$. The balance of a vector $v$ with only
nonnegative components is equal to $||v||_1$. For any slanted edge $[c]$ of a
corner cone, $c$ is the difference of two unit vectors, and thus
$\tau(c)=0$. For straight edges we have $\tau(c)=\tau(\pm e_i)=\pm 1$.

\begin{definition}[Visible Vectors]
  We call a vector $v\in\set R^m$ \emph{$i$-visible}, if
  \[(+,\dots+,\ominus,+\dots,+)\] is a valid sign vector for $v$, where $\ominus$
  is at the $i$th position.
\end{definition}

The definition is motivated by corner cones. For $i>0$, a vector is $i$-visible
iff it belongs to $\lat C_i$. An $i$-visible vector $v$ has nonpositive
balance $\tau(v)\leq 0$.

\begin{lemma}
  \label{lem:zero}
  Let $k\in\{1,\dots,m\}$ and let $v_1,\dots,v_k,v\in\set R^m$ be such that
  $v_1+v_2+\dots+v_k=v$ and that each $v_i$ lies in some corner cone. Suppose
  that there is an associated sign equation and indices $r_1,\dots,r_k$ such
  that a valid sign equation projected to rows $r_1,\dots,r_k$ is of the form
  \[
    \begin{bmatrix}\;
      \boxed{\ominus}\;\;
    \end{bmatrix}=
    \begin{pmatrix}\;
      \dboxed{+}\;\,
    \end{pmatrix}.
  \]
  Then for every $j\in\{r_1,\dots,r_k\}$, the $j$th component of $v$ is zero,
  and for every $j\in\{1,\dots,m\}\setminus\{r_1,\dots,r_k\}$, the $j$th
  component of $v_i$ is zero for every $i$.
\end{lemma}

\begin{proof}
  Let $\pi\colon\set R^m\rightarrow \set R^{n}$ be the projection on the components
  with indices $r_1,\dots,r_k$, and $\overline{\pi}$ the projection on the
  complementary components.
  The sign equation implies that $\tau(\pi(v_i))\leq 0$ for all $v_i$.
  It follows that $\tau(\pi(v))$ has
  to be less than or equal to $0$ as well. As $\pi(v)$ only contains
  nonnegative entries, this is only possible if $\pi(v)$ is the zero
  vector. This shows the first part and also implies the equation
  \[\tau(\pi(v_1)) + \tau(\pi(v_2)) + \dots + \tau(\pi(v_k))=0.\] Since no
  summand on the left hand side is strictly positive, all the $\tau(\pi(v_i))$ have to be
  equal to 0. As every $v_i$ lies in some corner cone, and their negative
  components only have indices contained in $\{r_1,\dots,r_k\}$, we get that all the
  $\overline{\pi}(v_i)$ only have nonpositive components. Now it follows
  that all $\overline{\pi}(v_i)$ are equal to zero, since
  \[0\geq\tau(\overline{\pi}(v_i))=\tau(\overline{\pi}(v_i))+\tau(\pi(v_i))
    =\tau(v_i)\geq 0.\]

  \vspace{-1.82\baselineskip}\qed
\end{proof}

\begin{remark}
  \label{remark1}
  Clearly, if a $v_1,\dots,v_k,v$ with $v_1+\dots+v_k=v$ are such that a valid
  sign equation contains rows of the form
  \[ \begin{bmatrix}\;
      \dboxed{-}\;\,
    \end{bmatrix}=
    \begin{pmatrix}\;
      \dboxed{+}\;\,
    \end{pmatrix},
  \]
  then $v$ and the $v_i$ can only contain zero entries at the corresponding indices.

\end{remark}
\begin{proposition}
  \label{prop:conedimension}
  Let $v_1,\dots,v_m\in\set R^m\setminus\{0\}$ be such that each $v_i$ is
  $i$-visible. If no subspace of $V:=\< v_1,\dots,v_m>$ can be decomposed into a
  direct sum of more than one nonzero vector spaces, then
  $\operatorname{dim}(V)= m-1.$
\end{proposition}
\begin{proof}
  If $V$ is of dimension $m$, it can be decomposed into the direct sum of $m$
  nonzero vector spaces.  Suppose that $\dim(V)<m-1$, and, without loss of
  generality, that $v_1,\dots,v_{m-2}$ generate $V$, i.e.\ there are
  $\alpha_1,\dots,\alpha_{m-2}$ and $\beta_1,\dots,\beta_{m-2}\in\set R$ with
  $v_{m-1}=\sum_{i<m-1} \alpha_iv_i$ and
  $v_m=\sum_{i<m-1}\beta_iv_i$.  For these we get corresponding sign equations
  \begin{align}
    \label{eq1}
    \pm\begin{pmatrix} \ominus\\ +\\ \vdots\\ +\\ +\\ +\\ +\end{pmatrix} \pm \dots \pm
    \begin{pmatrix} +\\ +\\ \vdots\\ +\\ \ominus\\ +\\ +\end{pmatrix} =
    \begin{pmatrix} +\\ +\\ \vdots\\  +\\ +\\ \ominus\\
      + \end{pmatrix},\\
    \label{eq2}
    \pm\begin{pmatrix} \ominus\\ +\\ \vdots\\ +\\ +\\ +\\ +\end{pmatrix} \pm \dots \pm
    \begin{pmatrix} +\\ +\\ \vdots\\ +\\ \ominus\\ +\\ +\end{pmatrix} =
    \begin{pmatrix} +\\ +\\ \vdots\\  +\\ +\\ +\\\ominus
    \end{pmatrix},
  \end{align}
  where the $\pm$ reflect the fact that the $\alpha_i$ and $\beta_i$ can be
  positive or negative.  We show that there is no combination of signs for the
  $\alpha_i$ and $\beta_i$ such that both equalities hold, unless $V$ can be
  decomposed into a direct sum. We first look at~\eqref{eq2}.  From the last row
  we see that at least one $\beta_i$ has to be strictly negative, as no $v_i$ on
  the left hand side has a negative entry at index $m$, and $v_{m}$ is
  nonzero. So we can split the vectors into two groups: those with positive
  $\beta_i$ and those with strictly negative $\beta_i$. After changing the
  summation order and reorganizing the rows if necessary, Equation~\eqref{eq2}
  becomes
  \begin{alignat*}2
    \begin{bmatrix}\raisebox{0pt}[12pt]{$\boxed{\ominus}$}\\[4pt] \dboxed{+}\\[4pt] +\;\cdots\;+\\[4pt]
      +\;\cdots\;+\end{bmatrix} +
    \begin{bmatrix}\raisebox{0pt}[12pt]{$\dboxed{-}$}\\[4pt] \boxed{\oplus}\\[4pt] -\;\cdots\;- \\[4pt]
      -\;\cdots\;-\end{bmatrix} & & {} =
    \begin{pmatrix} \raisebox{0pt}[12pt]{$\dboxed{+}$} \\[4pt] \dboxed{+} \\[4pt] +\\[4pt] \ominus \end{pmatrix}.
  \end{alignat*}
  Suppose at least one $\beta_i$ is strictly positive. Then Lemma~\ref{lem:zero}
  implies that some components have to be zero, and we get a block diagonal form
  \begin{alignat*}1
    \begin{bmatrix}
      \boxed{\ominus} \\[4pt]
      \dboxed{\hspace{1.4pt}0\hspace{1.4pt}}\\[4pt]
      0\;\cdots\;0\mathstrut\\[3pt]
      0\;\cdots\;0
    \end{bmatrix} +
    \begin{bmatrix}
      \dboxed{\hspace{1.4pt}$0$\hspace{1.4pt}}\\[4pt]
      \boxed{\oplus}\\[4pt]
      0\;\cdots\;0\\[3pt]
      -\;\cdots\;-
    \end{bmatrix} & =
    \begin{pmatrix} \raisebox{0pt}[12pt]{$\dboxed{+}$} \\[4pt] \dboxed{+} \\[4pt] 0\\[3pt] \ominus \end{pmatrix}.
  \end{alignat*}
  Thus, the $v_i$ appearing with a nonzero coefficient in~\eqref{eq1} span a
  vector space that can be decomposed into a direct sum if at least one
  $\beta_i$ is strictly positive. Otherwise, with the analogous reasoning for
  $v_{m-1}$, we can suppose that all the $\alpha_i$ and $\beta_i$ are
  nonpositive, and conclude that the sign equations for $v_{m-1}$ and $v_m$ are
  of the form
  \begin{align}
    \label{eq5a}
    \begin{bmatrix}\raisebox{0pt}[11pt]{$\boxed{\oplus}$}\\ -\;\cdots\;-\\ -\;\cdots\;- \end{bmatrix}  =
    \begin{pmatrix}\raisebox{0pt}[11pt]{$\dboxed{+}$}\\ \ominus\\ + \end{pmatrix},\\
    \label{eq5b}
    \begin{bmatrix}\raisebox{0pt}[11pt]{$\boxed{\oplus}$}\\ -\;\cdots\;-\\ -\;\cdots\;- \end{bmatrix}  =
    \begin{pmatrix}\raisebox{0pt}[11pt]{$\dboxed{+}$}\\ +\\ \ominus \end{pmatrix}.
  \end{align}
  If all $\alpha_i$ (implicitly used in~\eqref{eq5a}) were nonzero, then
  the last row in~\eqref{eq5a} implies that the last components of all the $v_i$ would have to
  be zero, which is incompatible with the last row in~\eqref{eq5b}. Again with
  the analogous reasoning for the $\beta_i$ we see that not all
  $\alpha_i$ and not all $\beta_i$ are nonzero. As before we split the vectors on the
  left hand side of each equation into two blocks: vectors that appear with a
  nonzero coefficient in only one of the equations and vectors that are shared
  in both equations with nonzero coefficients, which gives, after reordering
  the rows and summands if necessary:
  \begin{alignat*}2
    \begin{bmatrix}\raisebox{0pt}[12pt]{$\boxed{\oplus}$}\\[4pt] \dboxed{-}\\[4pt] \dboxed{-}\\[4pt]-\;\cdots\;-\\[4pt]-\;\cdots\;-\end{bmatrix} +
    & \begin{bmatrix}\raisebox{0pt}[12pt]{$\dboxed{-}$}\\[4pt] \boxed{\oplus} \\[4pt] \dboxed{-} \\[4pt]-\;\cdots\;-
      \\[4pt]-\;\cdots\;-\end{bmatrix}
    &{} =\begin{pmatrix}\raisebox{0pt}[12pt]{$\dboxed{+}$}\\[4pt] \dboxed{+}\\[4pt] \dboxed{ +}\\[4pt] \ominus\\[4pt] + \end{pmatrix},\\[4pt]
    & \underbrace{
      \begin{bmatrix}\raisebox{0pt}[12pt]{$\dboxed{-}$}\\[4pt] \boxed{\oplus} \\[4pt] \dboxed{-} \\[4pt]-\;\cdots\;- \\[4pt]-\;\cdots\;-\end{bmatrix}}_{\text{shared}}+
    \begin{bmatrix}\raisebox{0pt}[12pt]{$\dboxed{-}$}\\[4pt] \dboxed{-}\\[4pt] \boxed{\oplus}\\[4pt] -\;\cdots\;-\\[4pt]-\;\cdots\;-\end{bmatrix}
    &{} =\begin{pmatrix}\raisebox{0pt}[12pt]{$\dboxed{+}$}\\[4pt]\dboxed{+}\\[4pt] \dboxed{+}\\[4pt] +\\[4pt]\ominus \end{pmatrix}.
  \end{alignat*}
  We use Remark~\ref{remark1} to determine zero components
  in the first equation:
  \begin{alignat*}2
    \begin{bmatrix} \raisebox{0pt}[12pt]{$\boxed{\oplus}$}\\[4pt] \dboxed{-}\\[4pt] \dboxed{\hspace{1.4pt}0\hspace{1.4pt}}\\[4pt]-\;\cdots\;-\\[4pt]0\,\;\cdots\,\;0\end{bmatrix} +
    & \underbrace{\begin{bmatrix}\raisebox{0pt}[12pt]{$\dboxed{-}$}\\[4pt] \boxed{\oplus} \\[4pt] \dboxed{\hspace{1.4pt}0\hspace{1.4pt}} \\[4pt]-\;\cdots\;-
        \\[4pt]0\,\;\cdots\,\;0\end{bmatrix}}_{\text{shared}}
    &{} =\begin{pmatrix}\;\raisebox{0pt}[12pt]{$\dboxed{+}$}\;\;\;\\[4pt] \dboxed{+}\\[4pt] \dboxed{\hspace{1.4pt}0\hspace{1.4pt}}\\[4pt]
      \ominus\\[4pt] 0 \end{pmatrix}.
  \end{alignat*}
  Then, doing the same for the second equation, and using the fact that we
  already know some zero components in the shared vectors, we get:
  \begin{alignat*}2
    \underbrace{
      \begin{bmatrix}\raisebox{0pt}[12pt]{$\dboxed{\hspace{1.4pt}0\hspace{1.4pt}}$}\\[4pt] \boxed{\oplus} \\[4pt] \dboxed{\hspace{1.4pt}0\hspace{1.4pt}} \\[4pt]0\;\cdots\;0 \\[4pt]0\;\cdots\;0\end{bmatrix}}_{\text{shared}}+
    \begin{bmatrix} \raisebox{0pt}[12pt]{$\dboxed{\hspace{1.4pt}0\hspace{1.4pt}}$}\\[4pt] \dboxed{-}\\[4pt] \boxed{\oplus}\\[4pt] 0\;\,\cdots\,\;0\\[4pt]-\;\cdots\;-\end{bmatrix}
    &{} =\begin{pmatrix}\raisebox{0pt}[12pt]{$\dboxed{\hspace{1.4pt}0\hspace{1.4pt}}$}\\[4pt]\dboxed{+}\\[4pt] \dboxed{+}\\[4pt] +\\[4pt]\ominus \end{pmatrix}.
  \end{alignat*}
  Denote the number of shared vectors by~$s$. If $s$ is greater than $0$, we
  look at the rows in the equation for $v_{m-1}$ where the shared vectors are
  nonzero:
  \[
    \begin{bmatrix}
      \dboxed{-}
    \end{bmatrix} +
    \underbrace{\begin{bmatrix}
        \boxed{\oplus}
      \end{bmatrix}}_{\mathclap{s \text{ many
          shared}}} =
    \begin{pmatrix}
      \dboxed{+}
    \end{pmatrix}.
  \]
  As all nonshared vectors on the left hand side only have negative components,
  we can bring them to the right hand side and get:
  \[
    \underbrace{\begin{bmatrix}
        \boxed{\oplus}
      \end{bmatrix}}_{\mathclap{s \text{ many
          shared}}} =
    \begin{pmatrix}
      \dboxed{+}
    \end{pmatrix}.
  \]
  Note that here, all the hidden entries of the shared vectors are zero.  We can
  suppose that the shared vectors are linearly independent, otherwise we could replace some
  coefficients with zero. As they are linearly independent, however, they span
  the whole space $\set R^s$, thus the shared vectors can be replaced by unit
  vectors, which leads to a decomposition of $V$ into a direct sum of vector
  spaces. It remains to handle the case where there are no shared vectors
  in~\eqref{eq5a} and~\eqref{eq5b}. In that case, certain components
  in~\eqref{eq5a} and~\eqref{eq5b} have to be zero:
  \begin{alignat*}2
    \begin{bmatrix}\raisebox{0pt}[11pt]{$\boxed{\oplus}$}\\[4pt] \dboxed{\hspace{1.4pt}0\hspace{1.4pt}} \\[4pt] -\;\cdots\;- \\[4pt] 0\,\;\cdots\;\,0\end{bmatrix} & & {} =
    \begin{pmatrix}\;\raisebox{0pt}[11pt]{$\dboxed{+}$}\;\;\;\\[4pt] \dboxed{\hspace{1.4pt}0\hspace{1.4pt}} \\[4pt] \ominus\\[4pt] 0 \end{pmatrix},\\
    &\begin{bmatrix}\raisebox{0pt}[12pt]{$\dboxed{\hspace{1.4pt}0\hspace{1.4pt}}$}\\[4pt] \boxed{\oplus} \\[4pt]0\,\;\cdots\;\,0 \\[4pt]-\;\cdots\;-\end{bmatrix}  &{} =
    \begin{pmatrix}\;\raisebox{0pt}[12pt]{$\dboxed{\hspace{1.4pt}0\hspace{1.4pt}}$}\;\;\;\\[4pt] \dboxed{+}\\[4pt] 0\\[4pt] \ominus \end{pmatrix}.
  \end{alignat*}
  The zero entries on the left hand side imply that the space spanned by $V$ can
  be decomposed into a direct sum. This completes the proof.\qed
\end{proof}

\begin{corollary}
  \label{cor:nonsumdim}
  Let $\lat L\subseteq\set Z^m$, $m\geq 3$, be a lattice of dimension less than $m-1$ such that
  no subspace of the vector space spanned by $\lat L$ can be decomposed into a
  direct sum of two nonzero spaces. Then there exists a corner cone $\lat C$ of
  the standard simplex such that $\lat L\cap \lat C=\{0\}$.
\end{corollary}
\begin{proof}
  Let $v_i\in \lat L\cap \lat C_i$ for all $i=1,\dots,m$ and assume they are all
  nonzero. Then, Proposition~\ref{prop:conedimension} yields:
  \[\operatorname{dim}(\lat L)\geq \operatorname{dim}\< v_1,\dots,v_m>=m-1,\] a
  contradiction.\qed
\end{proof}

\begin{corollary}
  \label{cor:dimbound1}
  Let $\lat L=\< v_1,\dots,v_k>$ be a lattice in $\set Z^m$, $m\geq 3$. If
  $k=\dim(\lat L)<\frac{2}{3}m$, then there exists a corner cone
  $\lat C$ of the standard simplex such that $\lat L\cap \lat C=\{0\}$.
\end{corollary}
\begin{proof}
  Using Corollary~\ref{cor:nonsumdim} and projecting to the relevant coordinates
  shows that any subset $S$ of $v_1,\dots,v_k$ of some cardinality $s$ such that
  $S$ cannot be decomposed into a direct sum can only contain nonzero vectors
  of $s+1$ corner cones. Additionally for $s=1$, $S$ can only contain nonzero
  vectors of one corner cone. In fact, if there is a $v\in\{v_1,\dots,v_k\}$ and
  an $i\in\{1,\dots,m\}$ such that $v$ is $i$-visible, then it is immediate from
  the sign vector $(+,\dots+,\ominus,+\dots,+)$ of $v$ that
  $\<v>\cap \lat C_j=\{0\}$ for all $j\neq i$. It follows that $V$ can be
  decomposed into the sum of at most $\frac{k}{2}$ many two-dimensional vector
  spaces, each containing nonzero vectors of $3$ corner cones.\qed
\end{proof}

In order to derive a dimension bound such that both conditions in
part~\ref{prop:infinite:2} of Proposition~\ref{prop:infinite} are met, we need
the following lemma that allows us to construct a nonlonely point in a corner
cone from a nonlonely point in a different corner cone. A geometric interpretation of the statement is given in
Figure~\ref{fig:move}.

\begin{lemma}
  \label{lem:switchcones}
  Let $\lat C_i$ be a corner cone, $[c]$ be a slanted edge in $\lat C_i$, and
  let $j\in\set N$ be such that the $j$th component of $c$ is 1. If
  $\ell\in\lat L$ and $\alpha\in\set N$ are such that
  $v:=\ell + \alpha c\in\lat C_i$, then there exists a $\beta\in\set N^*$ with
  $\ell + \beta(-c)\in\lat C_j\setminus\{0\}$, where $[-c]$ is a slanted edge
  in the corner cone $\lat C_j$.
\end{lemma}

\begin{proof}
  By definition, the components of $c$ are all zero except for the $i$th
  component, which is $-1$, and the $j$th component for some $j\neq i$, which is
  $1$. Thus $[-c]$ is a slanted edge in $\lat C_j$. Set
  $\gamma:=\max(-v_i,\alpha)+1$, where $v_i$ is the $i$th component of $v$. Then
  $\tilde v:=v-\gamma c$ is $j$ visible, as $\tilde v_i = v_i+\gamma>0$,
  $\tilde v_k=v_k\geq 0$ for all $k\neq i,j$ and
  $\tilde v_j = v_j - \gamma \leq 0$ with
  \[-\tilde v_j = -v_j+\gamma = -v_j-v_i+(v_i+\gamma) >-v_j + \sum_{k\neq i} v_k
    + \tilde v_i=\sum_{k\neq j}\tilde v_k.\]
  Then, with $\beta:=\gamma-\alpha\in\set N^*$, we get  $\ell + \beta (-c) = \ell +
  (\alpha-\gamma) c = \tilde v\in \lat C_j$.\qed
\end{proof}

\begin{figure}
  \centering
    \begin{tikzpicture}[scale=0.5]
      \begin{scope}
        \clip (-4.5,-4.5) rectangle (4.5,4.5);
        \draw[dashed] (-5,5) -- (5,-5);
        \draw[color=darkred,->] (-2.99,2.995)--(-2.025,1.05);
        \draw[->] (0,0)--(-3,3);
        \fill (-2,1) circle (2pt);
        \node at (-1.2,2) {$\alpha c$};
        \node at (-2.8,1.7) {$\ell$};
        \node at (-2,0.6) {$v$};
      \end{scope}
      \draw[thick,<->] (-5,0) -- (5,0); \draw[thick,<->] (0,-5) -- (0,5);
      \draw[->] (4.5,1) to[in=135,out=45] (7.5,1);
      \node at (6,2) {$v\rightarrow \tilde{v}$};
      \begin{scope}
        \clip (7.5,-4.5) rectangle (16.5,4.5);
        \draw[dashed] (7,5) -- (17,-5);
        \draw[->] (12,0)--(13,-1);
        \draw[color=darkred,->] (13,-1)--(13.975,-2.95);
        \fill (14,-3) circle (2pt);
        \node at (13.8,-.45) {$ \beta (-c)$};
        \node at (13.2,-2.2) {$\ell$};
        \node at (14,-3.4) {$\tilde{v}$};

      \end{scope}
      \draw[thick,<->] (7,0) -- (17,0); \draw[thick,<->] (12,-5) -- (12,5);
    \end{tikzpicture}
  \caption{Illustration of Lemma~\ref{lem:switchcones} in dimension 2 with
    $c=(-1,1), \ell=(1,-2), v=(-2,1), \alpha=3, \gamma=4, \beta=1$ and $\tilde v
    = (2,-3).$}
  \label{fig:move}
\end{figure}
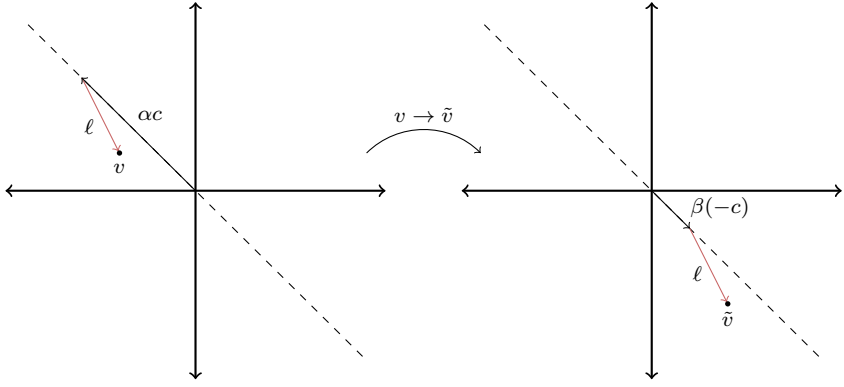
\begin{theorem}
  \label{thm:dimbound}
  Let $\lat L$ be a lattice in $\set Z^m$. If $\dim(\lat L)< \frac{1}{3}(m-4)$,
  then there exists a slanted edge $[c]$ in a corner cone $\lat C$ such that
  $\lat L\cap \lat C =\{0\}$ and $(\lat L+\<c>)\cap \lat C=[c]$.
\end{theorem}
\begin{proof}
  If $m\leq 4$, there is nothing to show. Suppose $m>4$ and, without loss of
  generality, even.  If there exist such
  $c$ and $\lat C$, then there is no nonzero $\ell\in\lat L$ and no nonzero
  $\alpha\in\set Z$ such that $\ell+\alpha c\in\lat C$.  Thus we can prove the claim
  by showing that if there exist $m(m-1)$ many nonzero
  $\ell_{1,2},\ell_{1,3}\dots,\ell_{2,1},\ell_{2,3},\dots,\ell_{m,m-1}\in\lat L$
  and nonzero $\alpha_{1,2},\dots,\alpha_{m,m-1}\in\set Z$ such that for each corner cone $\lat C_i$ and each slanted edge
  $[c_{i,j}]$ the vector
  $v_{i,j}:=\ell_{i,j}+\alpha_{i,j}c_{i,j}$ is $i$-visible, then
  $\dim(\lat L)\geq\frac{1}{3}(m-4)$. So suppose such $\ell_{i,j}$ and
  $\alpha_{i,j}$ exist. Then $(\dboxed{+},\ominus,\dboxed{+})$ is a valid sign
  vector for $v_{i,j}$, with $\ominus$ at the $i$th position. We first show that
  each $\ell_{i,j}$ is either $i$-visible, $j$-visible, or has exactly two
  strictly negative entries, at indices $i$ and $j$. For the moment, we focus on
  $i=1, j=2$, allowing us to drop both indices. The reasoning for all other
  pairs $i,j$ is analogous.  We get the equation:
  $\ell + \alpha c = v.$

  \noindent If $\alpha\leq 0$, we can add $-\alpha c$ to both
  sides of the equation, not perturbing the $1$-visibility of the right
  hand side, which shows that $\ell$ is $1$-visible. Otherwise, we get a sign
  equation with unknown entries for $\ell$,
  \[
    \substack{\begin{pmatrix}
        \raisebox{0pt}[11pt]{$?$}\\[4pt]
        ?\\[4pt]
        \raisebox{0pt}[0pt][8pt]{$\dboxed{\hspace{1.4pt}?\hspace{1.4pt}}$}
      \end{pmatrix}\\ \ell} +
    \substack{\begin{pmatrix}
        \raisebox{1pt}[11pt]{$\ominus$}\\[4pt]
        \raisebox{2.5pt}{$\oplus$}\\[4pt]
        \raisebox{0pt}[0pt][8pt]{$\dboxed{\hspace{1.4pt}0\hspace{1.4pt}}$}
      \end{pmatrix}\\ \alpha c} =
    \substack{\begin{pmatrix}
        \raisebox{0pt}[11pt]{$\ominus$}\\[4pt]
        \raisebox{.5pt}{$+$}\\[4pt]
        \raisebox{0pt}[0pt][8pt]{$\dboxed{+}$}
      \end{pmatrix}\\ v}.\]
  The signs for all but two components of $\ell$ are immediate:
  \[
    \begin{pmatrix}
      \raisebox{0pt}[11pt]{$?$}\\[4pt]
      ?\\[4pt]
      \raisebox{0pt}[0pt][8pt]{$\dboxed{+}$}
    \end{pmatrix} +
    \begin{pmatrix}
      \raisebox{1pt}[11pt]{$\ominus$}\\[4pt]
      \raisebox{2.5pt}{$\oplus$}\\[4pt]
      \raisebox{0pt}[0pt][8pt]{$\dboxed{\hspace{1.4pt}0\hspace{1.4pt}}$}
    \end{pmatrix} =
    \begin{pmatrix}
      \raisebox{0pt}[11pt]{$\ominus$}\\[4pt]
        \raisebox{.5pt}{$+$}\\[4pt]
      \raisebox{0pt}[0pt][8pt]{\dboxed{+}}
    \end{pmatrix}.
  \]
  As $\tau(c)=0$ and $\tau(v)\leq 0$, we get that $\tau(\ell)\leq
  0$. Thus, if the second component of $\ell$ is positive, then it follows that
  $\ell$ is $1$-visible with a strictly negative first component. If the second
  component of $\ell$ is negative, we can apply Lemma~\ref{lem:switchcones} to
  see that there exists a $\beta\in\set N^*$ such that
  $\ell+\beta(-c)$ is $2$-visible, yielding
  \[
    \begin{pmatrix}
      \raisebox{0pt}[11pt]{?}\\[4pt]
      -\\[4pt]
      \raisebox{0pt}[0pt][8pt]{\dboxed{+}}
    \end{pmatrix} +
    \begin{pmatrix}
      \raisebox{1pt}[11pt]{$\oplus$}\\[4pt]
      \raisebox{0pt}{$\ominus$}\\[4pt]
      \raisebox{0pt}[0pt][8pt]{\dboxed{\hspace{1.4pt}0\hspace{1.4pt}}}
    \end{pmatrix} =
    \begin{pmatrix}
      \raisebox{0pt}[11pt]{$+$}\\[4pt]
      \ominus\\[4pt]
      \raisebox{0pt}[0pt][8pt]{\dboxed{+}}.
    \end{pmatrix}
  \]
  With the same reasoning as above we can determine that $\ell$ is either
  $2$-visible or its first component is strictly negative. This shows our claim
  for the $\ell_{i,j}$. It follows that for each pair $(i,j)$, the vector
  $\ell_{i,j}$ is such that it has a strictly negative entry at $i$, or $j$, or
  both. Thus we can find at least $\frac m 2$ pairwise different
  $\ell_1,\dots,\ell_{\frac m 2}\in\lat L$ such that no two $\ell_i$ have a
  negative entry at the same index, and for each index in $\{1,\dots,m\}$, there
  is exactly one $\ell_i$ with a negative entry at that position.
  We now map these lattice elements to $i$-visible vectors,
  $i=1,\dots,\frac m 2-2$, in $\set Z^{m/2}$.

  \noindent
  For any permutation $\pi$ of $1,\dots,m$ consider the surjective linear map
  \begin{alignat*}1
    & \psi_\pi:\set R^m\rightarrow \set R^{m/2},\\
    & (u_1,\dots,u_m)\mapsto
    (u_{\pi(1)}+u_{\pi(2)},u_{\pi(3)}+u_{\pi(4)},\dots,u_{{\pi(m-1)}}+u_{\pi(m)}).
  \end{alignat*}
  There are $n(m):=m!/2^{m/2}$ many such maps.  We say a vector $u$
  and a map $\psi_\pi$ are compatible, if:
  \begin{itemize}
  \item $u$ is $i$-visible for some $i$, and $\psi_\pi(u)\neq 0$. If
    $a\in\set N$ is such that $\pi(a)=i$, then $\psi_\pi$ is
    $\lfloor \frac{a+1}{2}\rfloor$-visible.
  \item $u$ contains exactly two strictly negative entries at indices $i$ and
    $j$, and there is an odd integer $a$ such that $\pi(a)=i$ and $\pi(a+1)=j$,
    i.e.\ when applying $\psi_\pi$ on $v$, the two negative entries are added
    together to give an $\frac{a+1}{2}$-visible vector.
  \end{itemize}
  We now show that there exists a permutation $\pi$ such that at least
  $\frac m 2 -2$ many $\ell_i$ are compatible to $\psi_\pi$. In fact we can choose
  $\pi$ such that all $\ell_i$ with exactly two negative entries are compatible with
  $\psi_\pi$, as they do not have negative entries at the same indices. This
  leaves us with some even number $k\geq 0$ of indices not yet considered for
  $\pi$ and $k$ many $\ell_i$ that could potentially be incompatible to such a
  permutation. Furthermore, there are $n(k)$ many permutations left to choose
  from.  Each of the remaining $\ell_i$ is contained in a different corner cone,
  say $\lat C_i$, and so $\ell_i$ is incompatible if $\psi_\pi(\ell_i)=0$. For $k> 2$,
  each $\ell_i$ can be in the kernel of at most $n(k-2)$ many of the remaining
  permutations (this is the case if $\ell_i$ is contained in a slanted edge of a
  corner cone). As there are $k$ ($k>2$, even) many such $\ell_{i}$, there has to be a
  $\psi_\pi$ for which the number of $i$-visible $\ell_i$ that are mapped to zero
  is at most
 \[\left\lfloor k \frac{n(k-2)}{n(k)}\right\rfloor =
   \left\lfloor\frac{2}{k-1}\right\rfloor=0.\] For $k=2$, there is only one
 choice for $\pi$, and we could be in the situation where both of the $\ell_i$
 have to be mapped to zero.  For any such $\pi$, the images of the $\ell_{i}$
 therefore contain at least $m/2-2$ many nonzero vectors with $m/2-2$ different
 sign patterns (after potentially reordering the rows)
  \[
    \begin{pmatrix}
      \ominus\\
      +\\
      \vdots\\
      +\\
      +\\
      +
    \end{pmatrix},
    \begin{pmatrix}
      +\\
      \ominus\\
      \vdots\\
      +\\
      +\\
      +
    \end{pmatrix},\dots,
    \begin{pmatrix}
      +\\
      +\\
      \vdots\\
      \ominus\\
      +\\
      +
    \end{pmatrix}.
  \]
  By projecting to the first $\frac m 2 -2$ coordinates and using
  Corollary~\ref{cor:dimbound1}, it follows that
  \[\dim\<\ell_{1},\dots,\ell_{\frac m
      2}>=\dim\<\psi_\pi(\ell_{1}),\dots,\psi_\pi(\ell_{\frac m 2})>\geq
    \frac{2}{3}\bigl(\frac{m}{2}-2\bigr)=\frac{1}{3}(m-4).\]
  This proves the claim.\qed
\end{proof}

  Without further restrictions on $\lat L$, there is no analogous result for
  straight edges, i.e.\ there is no upper bound for the dimension proportional to $m$
  such that lower dimensional lattices necessarily lead to infinitely many
  lonely points on at least one straight edge. For any $m$, the lattice
  generated by $(1,0,\dots,0)$ yields only finitely many lonely points on any
  straight edge.

\section{Conclusion and Open Questions}
\label{sec:conclusion}

We translated the problem of reducing the order of a C-finite sequence to
questions about which points in a dilated simplex are not connected to any other
point in the simplex via a specific lattice. Our answers to these questions are
in the form of algorithms that determine when the number of these points grows
indefinitely with the dilation, and also compute the exact number if there are
only finitely many lonely points. Furthermore we showed that if the
dimension of the lattice is small enough, then the number of lonely points
always grows indefinitely.

Theorem~\ref{thm:dimbound} is helpful for our original application to C-finite
sequences, because the lattices appearing in this context are typically
small. We do not know however whether the bound of Theorem~\ref{thm:dimbound} is
tight enough to cover all cases of interest. If it is not, we can still use the
Algorithms from Section~\ref{sec:algs} to see whether there are enough lonely
points to derive a finite degree bound for the ansatz.  At the moment, we do not
know whether this is always the case.

As for extensions of our theoretical results, there are immediate questions that
are rooted in discrete geometry: Can we find a closed form expression depending
on $d$ for the number of lonely points in $d\lat S$ for a given lattice? How
many lonely points are there in more involved convex polytopes? How do linear
transformations on the lattice affect lonely points? While these kinds of
questions are more removed from the initial number theoretic application, their
pursuit may lead to valuable insight.

\bibliographystyle{spmpsci}
\bibliography{references}

\end{document}